\documentclass[showpacs,amsmath,amssymb,twocolumn,prx,superscriptaddress,10pt,aps]{revtex4-2}

\usepackage{qcircuit}
\usepackage[dvips]{graphicx}
\usepackage{siunitx}
\usepackage{amsmath,amssymb,amsthm,mathrsfs,amsfonts,dsfont,pifont}
\usepackage{subfigure, epsfig}
\usepackage{bm}
\usepackage{enumerate}
\usepackage[dvipsnames]{xcolor}
\usepackage{color}
\usepackage{fancybox, graphicx}
\usepackage[skins]{tcolorbox}
\usepackage{multirow}
\usepackage{mathtools}
\usepackage{thmtools}
\usepackage{thm-restate}
\usepackage{makecell}
\usepackage{mleftright}\mleftright 

\DeclarePairedDelimiter\bra{\langle}{|}
\DeclarePairedDelimiter\ket{|}{\rangle}
\DeclarePairedDelimiterX\braket[2]{\langle}{\rangle}{#1 \delimsize\vert #2}
\DeclarePairedDelimiterX\ketbra[2]{| }{|}{#1 \delimsize\rangle\!\delimsize\langle #2}
\newcommand{\vertiii}[1]{{\left\vert\kern-0.25ex\left\vert\kern-0.25ex\left\vert #1 
		\right\vert\kern-0.25ex\right\vert\kern-0.25ex\right\vert}}
\newcommand{\bigO}[1]{\mathcal{O}\left( #1 \right)}

\usepackage[colorlinks]{hyperref}
\hypersetup{
	draft=false,
	colorlinks=true,
	citecolor = {blue},
    linkcolor=blue,
    filecolor=magenta,      
    urlcolor=cyan,
}
\usepackage{cleveref}

\usepackage{ifdraft}
\ifdraft{
 \newcommand{\authnote}[3]{{\color{#3} {\bf  #1:} #2}}
 }{
 \newcommand{\authnote}[3]{}
 }

\newtcolorbox[auto counter]{tbox}[2][]{%
	enhanced, float=hbt, drop fuzzy shadow southeast,
	colback=white!5!white, colframe=white!30!black,
	width= .97\columnwidth,sharp corners,boxrule=0.8pt,
	title={#2}, #1
}

\newtheorem{theorem}{Theorem}
\newtheorem{lemma}{Lemma}
\newtheorem{corollary}{Corollary}

\newcommand{\mx}[1]{|{#1}|_\mathrm{max}}
\newcommand{\twonm}[1]{||{#1}||_2}
\newcommand{\fillfrac}[2]{\mathcal{F}_{#1}^{[{#2}]}}
\newcommand{\norm}[1]{\mathcal{N}_{#1}}
\newcommand{\eps}{\varepsilon}
\newcommand{\fu}{\underline{f}}
\newcommand{\fut}{\underline{\tilde{f}}}

\newtcolorbox{codebox}{enhanced, width=.95\columnwidth, halign = flush left, drop fuzzy shadow southeast, boxrule=0.4pt, sharp corners, colframe=black, colback=white}

\begin{document}

\title{Quantum state preparation without coherent arithmetic}

\author{Sam McArdle}
\affiliation{AWS Center for Quantum Computing, Pasadena, CA 91125, USA}

\author{Andr\'{a}s Gily\'{e}n}
\affiliation{Alfr\'{e}d R\'{e}nyi Institute of Mathematics, Budapest, Hungary}

\author{Mario Berta}
\affiliation{AWS Center for Quantum Computing, Pasadena, CA 91125, USA}
\affiliation{Department of Computing, Imperial College London, London, UK}
\affiliation{Institute for Quantum Information, RWTH Aachen University, Aachen, Germany}

\date{\today}


\begin{abstract}
We introduce a versatile method for preparing a quantum state whose amplitudes are given by some known function. Unlike existing approaches, our method does not require handcrafted reversible arithmetic circuits, or quantum table reads, to encode the function values. Instead, we use a template quantum eigenvalue transformation circuit to convert a low cost block encoding of the sine function into the desired function. Our method uses only $4$ ancilla qubits (3 if the approximating polynomial has definite parity),  providing order-of-magnitude qubit count reductions compared to state-of-the-art approaches, while using a similar number of gates if the function can be well represented by a polynomial or Fourier approximation. Like black-box methods, the complexity of our approach depends on the `L2-norm filling-fraction' of the function. We demonstrate the algorithmic utility of our method, including preparing Gaussian and Kaiser window states.
\end{abstract}

\maketitle


\section{Introduction}

\paragraph*{Problem setting.} We seek to prepare an $N=2^n$ dimensional quantum state on $n$ qubits with amplitudes described by a known function $f(\bar{x})$ (where $\bar{x}$ is a suitable rescaling of the binary qubit register state $\ket{x}$). Such states are used in many quantum algorithms, including: basis and boundary functions in finite element analysis~\cite{montanaro2016quantum, scherer2017LinearSystemsResource} or differential equations~\cite{berry2017quantum,leyton2008quantum,cao2013quantum}, states in quantum simulations of field theories~\cite{jordan2012quantumfieldtheory,klco2021HierarchicalFieldTheory}, payoff and price distribution functions for financial derivative pricing~\cite{stamatopoulos2020option, Chakrabarti2021thresholdquantum}, priors for phase estimation~\cite{berry2022quantifyingTDA}, and radial and angular electron-orbital wave-functions in grid-based quantum chemistry simulations~\cite{ward2009preparation,chan2022grid}. Typical preparation methods~\cite{grover2000synthesis,sanders2019black,wang2021fast,rattew2022preparing} require an amplitude oracle $\ket{x} \ket{0} \rightarrow \ket{x} \ket{f(\bar{x})}$ that prepares a $g$-bit approximation of $f(\bar{x})$ (or some closely related oracle~\cite{grover2002creating,bausch2020fast, wang2021Inverseblackbox}). This can be implemented either by coherent arithmetic~\cite{munoz2018t, bhaskar2015quantum,haner2018optimizing}, or by reading values stored in a quantum lookup-table~\cite{ScirateThread2022, krishnakumar2022lookuptable}. Both can have high qubit and gate costs. Coherent arithmetic circuits are manually-optimized to minimize resources and incorporate the nuances of fixed-point arithmetic, such as overflow errors~\cite{haner2018optimizing}. Our approach does not use an amplitude oracle, saving considerable resources. This is vital in the early fault-tolerant regime, where we seek to minimize the footprint of quantum algorithms~\cite{campbell2021early,wan2022randomized,lin2022heisenberggroundstate,dong2022ground}.


\paragraph*{Framework.} Our method uses quantum singular value transformation (QSVT)~\cite{gilyen2018quantum} a technique to coherently apply functions to the singular values of a block-encoded matrix~\footnote{In this work, we block-encode a diagonal Hermitian matrix. The singular values of this matrix are the absolute values of the eigenvalues. Thus QSVT will perform eigenvalue transformation, where the sign information is stored in the left singular vectors.}. An $(n+m)$-qubit unitary $U$ is said to be an $(\alpha, m, \epsilon)$-block-encoding of an $n$-qubit Hermitian matrix $A$ if
\begin{equation}
    \bigg{|}\bigg{|} \alpha \left( \bra{0}^{\otimes m} \otimes I_n \right) U \left( \ket{0}^{\otimes m} \otimes I_n \right) - A \bigg{|}\bigg{|} \leq \epsilon.
\end{equation}
The default QSVT approach~\cite{gilyen2018quantum} uses $d/2$ applications each of $U, U^\dagger$, $2d$  $m$-controlled Toffoli gates (which are just CNOT gates for the $m=1$ case herein), and $\mathcal{O}(d)$ single-qubit gates to block-encode a degree $d$ real and definite-parity polynomial of $A$. Using linear combinations of block-encodings, we can block-encode complex, mixed-parity functions~\cite{gilyen2018quantum}. 


\paragraph*{Approach.} We present our method in detail for $f\colon[-a, a]\rightarrow \mathbb{R}$ of definite-parity, and seek to prepare 
    \begin{equation*}
    \ket{\Psi_f} := \frac{1}{\norm{f}} \sum_{x=-\frac{N}{2}}^{\frac{N}{2}-1} f\left( \bar{x} \right) \ket{x},
\end{equation*}
where $N=2^n$, $\bar{x} :=\left( 2ax/N\right)$, and $\mathcal{N}_f := \sqrt{\sum |f(\cdot)|^2}$. We use a two's complement representation of signed integers (see Appendix~\ref{App:SignedIntegers})\footnote{The method can be easily adapted to other representations of integers.}.
The extension to the mixed-parity and complex case can be achieved through linear combinations of block-encodings~\cite{gilyen2018quantum,dong2021efficient}. As shown in Fig.~\ref{fig:CircuitDiagram}, we use QSVT to convert a low-cost block-encoding of $A:=\sum_{x=-\frac{N}{2}}^{\frac{N}{2}-1}\sin(2x/N) \ketbra{x}{x}$, into a block-encoding of $\sum_x f(\bar{x}) \ketbra{x}{x}$, using a polynomial approximation of $f(a\arcsin(\cdot))$. Our approach is well suited to functions with low-degree polynomial (or Fourier) approximations, and provides order-of-magnitude reductions in the number of ancilla qubits used. Unlike amplitude-oracle-based approaches, we avoid discretizing the values the function can take, yielding a continuous approximation to the function. Our method is versatile, as the same circuit template can be used for all functions.

\paragraph*{Related work.} Refs.~\cite{van2019quantumzerosum,guo2021nonlinear} used similar QSVT-based techniques for a related task of transforming amplitudes encoded via a black-box state-preparation unitary or QRAM. If used for the task considered herein, these techniques would require more qubits and introduce a larger subnormalization factor than our white-box approach.

\paragraph*{Outline.} Sec.~\ref{Sec:MainResult} introduces our method, with our main result presented in \Cref{thm:genericComplexity}. Sec.~\ref{Sec:Applications} provides theoretical complexities and concrete resource estimates for preparing algorithmically valuable functions. Sec.~\ref{Sec:Extensions} discusses extensions for dealing with discontinuities, using improved priors, and Fourier approximations.


\section{Main result}\label{Sec:MainResult}

For a function $p(y)$ in the range $y \in [-a,a]$ we define the `discretized L2-norm filling-fraction'
\begin{equation}
    \fillfrac{p}{N} = \frac{\norm{p}}{\sqrt{N} |p(y)|_{\mathrm{max}}^{y \in [-a,a]}}
\end{equation}
which approximates the continuous quantity $\fillfrac{p}{\infty} := \sqrt{ \frac{\int_{-a}^a |p(y)|^2 dy}{2a \left(|p(y)|_{\mathrm{max}}^{y \in [-a,a]} \right)^2} }$. This quantity plays a key role in the complexity of our state preparation technique. 

Our method also requires a degree $d$ definite-parity polynomial $h(y)$, obeying $|h(y)|_{\mathrm{max}}^{y \in [-1,1]} \leq  1$, such that $\tilde{f}(y) := h(\sin(y/a))$ approximates the definite-parity function $f(y)$ on the interval $[-a,a]$. Given a sufficiently good $h(\cdot)$, we prove the following main result:

\begin{restatable}{theorem}{General}\label{thm:genericComplexity}
Given a degree $d$ definite-parity function $h(y)$ such that $|h(y)|_{\mathrm{max}}^{y \in [-1,1]} \leq  1$, which approximates $f(\cdot)$ as 
\begin{equation}
    \left|\tilde{f}(y) - \frac{f(a y)}{\mx{f(ay)}^{y \in [-1,1]}}\right|_{\mathrm{max}}^{y \in [-1,1]}  \leq \frac{\epsilon~\cdot~\mathrm{Min}\left(\fillfrac{f}{N},  \fillfrac{\tilde{f}}{N} \right)}{3} 
\end{equation}
where $\tilde{f}(y) := h(\sin(y/a))$, then we can prepare a quantum state $\ket{\Psi_{\tilde{f}}}$ that is no more than $\epsilon$-far from $\ket{\Psi_f}$ in trace distance using a quantum circuit requiring $\mathcal{O}\left( \frac{n d}{\fillfrac{\tilde{f}}{N}}  \right)$ gates and at most 3 ancilla qubits.
\end{restatable}
\begin{proof}
    A full proof is given in Appendix~\ref{App:ErrorAnalysis}. We sketch the main proof idea here. Recall $\bar{x} = 2ax/N$. The circuit in Fig.~\ref{fig:CircuitDiagram}a implements a $(1,1,0)$ block-encoding $U_{\mathrm{sin}}$ of $\sum_x \mathrm{sin}(\bar{x}/a) \ket{x}\bra{x}$ using $\mathcal{O}(n)$ gates. The circuit in Fig.~\ref{fig:CircuitDiagram}b uses QSVT to implement a $(1,2,0)$ block-encoding $U_{\tilde{f}}$ of $\sum_x h(\mathrm{sin}(\bar{x}/a)) \ket{x}\bra{x} = \sum_x \tilde{f}(\bar{x}) \ket{x}\bra{x}$ using $\mathcal{O}(d)$ calls to $U_{\sin}$, $U_{\mathrm{sin}}^\dag$ and $\mathcal{O}(d)$ additional elementary gates. The requirement $|h(y)|_{\mathrm{max}}^{y \in [-1,1]} \leq  1$ ensures the polynomial can be applied as a QSVT transformation. Applying $U_{\tilde{f}}$ to $\ket{00}\frac{1}{\sqrt{N}} \sum_x \ket{x}$ and measuring the ancilla qubits in $\ket{00}$ outputs $\ket{\Psi_{\tilde{f}}}$ that is no more than $\epsilon$-far from $\ket{\Psi_f}$ in trace distance with success probability at least $\frac{4}{9} \left(\fillfrac{\tilde{f}}{N}\right)^2$. The circuit in Fig.~\ref{fig:CircuitDiagram}c applies exact amplitude amplification (see Appendix~\ref{App:AmpAmp}) to boost the success probability to unity, using $\mathcal{O}\left(1/\fillfrac{\tilde{f}}{N}\right)$ calls to $U_{\tilde{f}}$, $U_{\tilde{f}}^\dag$, $\mathcal{O}\left(n/\fillfrac{\tilde{f}}{N}\right)$ additional elementary gates, and at most one additional ancilla qubit. In total, the circuit uses $\mathcal{O}\left( \frac{n d}{\fillfrac{\tilde{f}}{N}}  \right)$ gates and at most 3 ancilla qubits.
\end{proof}

\begin{figure}[!ht]
    \hskip-2mm
    \scalebox{0.8}{
    \begin{minipage}{11cm}
    	\hskip-113mm a)\vskip-10mm
		\[
		\Qcircuit @C=.5em @R=0.2em @!R {
			\lstick{\ket{a_1}} & \push{\rule{.1em}{0em}} & \gate{H} & \ctrl{4} & \qw & \qw & \ctrl{4} & \qw & \gate{R_z(\phi)} & \gate{H} & \gate{Y} & \qw \\
			\lstick{\ket{x_0}} & \push{\rule{.1em}{0em}} & \qw & \targ & \qw & \gate{R_z\left(2^{1-n}\right)} & \targ & \qw & \qw & \qw & \qw \\
			\lstick{\ket{x_1}} & \push{\rule{.1em}{0em}} & \qw & \targ & \qw & \gate{R_z\left(2^{2-n}\right)} & \targ & \qw & \qw & \qw & \qw \\
			\lstick{\vdots} & \push{\rule{.1em}{0em}} & \qw & \targ & \qw & \gate{\vdots} & \targ & \qw & \qw & \qw & \qw  \\
			\lstick{\ket{x_{n-1}}} & \push{\rule{.1em}{0em}} & \qw & \targ & \qw & \gate{R_z\left(-2^{0}\right)} & \targ & \qw & \qw & \qw & \qw
		}
		\]
    	
    	\hskip-113mm b)\vskip-10mm
		\[
		\Qcircuit @C=.3em @R=0em @!R {
			\lstick{\ket{a_2}} & \push{\rule{0mm}{7mm}}\qw & \gate{H} & \targ & \gate{R_z^{\theta_1}} & \targ & \qw & \qw & \targ & \gate{R_z^{\theta_2}} & \targ & \qw & \qw & \qw & \push{\rule{.1em}{0em}\dots\rule{.1em}{0em}} \\
			\lstick{\ket{a_1}} & \qw & \multigate{1}{U_{\mathrm{sin}}} & \ctrlo{-1} & \qw & \ctrlo{-1} & \multigate{1}{U_{\mathrm{sin}}^\dag} & \qw & \ctrlo{-1} & \qw & \ctrlo{-1} & \qw & \multigate{1}{U_{\mathrm{sin}}} & \qw & \push{\rule{.1em}{0em}\dots\rule{.1em}{0em}} \\
			\lstick{\ket{x}_n } & {/} \qw &  \ghost{U_{\mathrm{sin}}} & \qw & \qw & \qw & \ghost{U_{\mathrm{sin}}^\dag} & \qw & \qw & \qw & \qw & \qw & \ghost{U_{\mathrm{sin}}} & \qw & \push{\rule{.1em}{0em}\dots\rule{.1em}{0em}}
		}
		\]
    	
    	\hskip-113mm c)\vskip-10mm
		\[
		\Qcircuit @C=0.2em @R=.4em {
			\lstick{\ket{0}_{a_3}} & \qw & \qw & \gate{R_y(\omega)} & \ctrlo{1} & \gate{R_y(-\omega)} & \qw & \ctrlo{1} & \qw & \gate{R_y(\omega)} & \qw &\push{\rule{.1em}{0em}\dots\rule{.1em}{0em}}\\ 
			\lstick{\ket{00}_{a_1 a_2} } & {/} \qw &  \qw & \multigate{1}{U_{\tilde{f}}} & \ctrlo{-1} & \multigate{1}{U_{\tilde{f}}^\dag} & \qw & \ctrlo{-1} & \qw & \multigate{1}{U_{\tilde{f}}} & \qw & \push{\rule{.1em}{0em}\dots\rule{.1em}{0em}}\\ 
			\lstick{\ket{\bar{0}}_n} & {/} \qw & \gate{H^{\otimes n}} & \ghost{U_{\tilde{f}}} & \qw & \ghost{U_{\tilde{f}}^\dag} & \gate{H^{\otimes n}} & \ctrlo{-1} & \gate{H^{\otimes n}} & \ghost{U_{\tilde{f}}} &\qw & \push{\rule{.1em}{0em}\dots\rule{.1em}{0em}}
		}
		\]
    \end{minipage}
	}\kern-9mm
    \caption{The quantum circuit implementing QSVT-based state preparation. We define $R_y(\theta) := e^{-i\theta Y}$, $R_z(\theta) = \mathrm{Diag}(1, e^{i\theta})$.
    a) The circuit $U_{\mathrm{sin}}$ that block-encodes $\sum_{x} \mathrm{sin}(2x/N) \ketbra{x}{x}$ by applying a Hadamard test circuit to a directionally controlled phase gradient~\cite{GidneyBlog2017} (see Lemma~\ref{Lemma:BlockEncodeSin}). This circuit requires (n+1) $Z$ rotations, and CNOT chains that can be implemented in $\mathcal{O}(\log(n))$ depth~\cite{low2018trading}, and can be further optimized for fault-tolerant implementation in e.g. the surface code~\footnote{When implementing the multitarget CNOT gates using lattice surgery, they can be implemented in depth independent of $n$. The $Z$ rotations (which must be decomposed into a number of $T$ gates) can be replaced by an addition circuit composed of Toffoli gates by using a phase gradient catalyst state~\cite{Gidney2018halvingcostof,litinski2022active}.}. b) The circuit $U_{\tilde{f}}$ that block-encodes $\sum_x \tilde{f}(\bar{x}) \ketbra{x}{x} $ by applying QSVT to $U_{\mathrm{sin}}$. The angles $\theta_i$ correspond to the pre-computed QSVT-angles for the desired polynomial. c) The (exact) amplitude-amplification circuit which block encodes $ \ketbra{\Psi_{\tilde{f}}}{\bar{0}}$, including an additional qubit to adjust the amplitude (see Appendix~\ref{App:AmpAmp}).}
    \label{fig:CircuitDiagram}
\end{figure}

The constant factor hidden by the big-$\mathcal{O}$ notation is function dependent, and may depend on the scaling factor $a$. For smooth functions that can be well approximated by polynomials, one can typically obtain an $L_\infty$-error $\delta$ decaying as $\mathcal{O}\left( \exp(-d) \right)$ for a degree $d$ approximating polynomial. We prove this formally in Appendix~\ref{App:Taylor}. For such functions, we can then prepare a quantum state $\ket{\Psi_{\tilde{f}}}$ that is $\epsilon$-close in trace-distance to $\ket{\Psi_f}$ 
using 
\begin{equation}
    \widetilde{\mathcal{O}}\left( \frac{n}{\fillfrac{\tilde{f}}{N}} \log\left( \frac{1}{\epsilon}\right) \right)
\end{equation}
gates, where the notation $\widetilde{\mathcal{O}}(\cdot)$ hides poly-logarithmic terms. As $N$ is increased, $\fillfrac{\tilde{f}}{N} \rightarrow \fillfrac{\tilde{f}}{\infty}$, a constant value independent of $N$, for a given function. Furthermore, in practice the error analysis can be tightened, as discussed in Appendix~\ref{App:TighterErrorAnalysis}.

\begin{table*}[t]
    \centering
    \begin{tabular}{c|c|c|c|c}
    & \makecell{\# Calls to amplitude oracle} & \makecell{\# Non-Clifford gates} & \makecell{\# Ancilla qubits} & Applicability \\ \hline
        \makecell{QSVT-based (This work)} & None & $\mathcal{O}\left( n  d_{\epsilon}/\fillfrac{\tilde{f}}{N}  \right)$ & 3 &\makecell{Polynomial approximation} \\ \hline
\makecell{Black-box \cite{grover2000synthesis,sanders2019black,bausch2020fast,wang2021fast,wang2021Inverseblackbox}}   & $\mathcal{O}\left( 1/ \fillfrac{f}{N}  \right)$ & $\mathcal{O}\left( g_{\epsilon}^2 \tilde{d}_{\epsilon} / \fillfrac{f}{N}  \right)$ & $\mathcal{O}(g_{\epsilon}\tilde{d}_{\epsilon})$ & \makecell{Generally applicable} \\ \hline
Grover-Rudolph~\cite{grover2002creating} & $\mathcal{O}\left(n \right)$  &  $\mathcal{O}\left(n g_{\epsilon}^2 \tilde{d}_{\epsilon} \right)$  & $\mathcal{O}(g_{\epsilon} \tilde{d}_{\epsilon})$ & \makecell{Efficiently integrable \\probability distributions} \\ \hline
\makecell{Adiabatic state\\ preparation~\cite{rattew2022preparing}} & $\mathcal{O}\left( \frac{1}{\left(\fillfrac{f}{N} \right)^4 \epsilon^2} \right)$ & $\mathcal{O}\left( \frac{g_{\epsilon}^2 \cdot \tilde{d}_{\epsilon}}{\left(\fillfrac{f}{N} \right)^4 \epsilon^2} \right)$ & $\mathcal{O}(g_{\epsilon} \tilde{d}_{\epsilon})$ & \makecell{Generally applicable}
    \end{tabular}
    \caption{Comparison of preparing real, definite parity $\ket{\Psi_f}$. $X_{\epsilon}$ indicates that $X$ depends on the error $\epsilon$. We instantiate $g_\epsilon$-bit amplitude oracles using the coherent arithmetic approaches of~\cite{haner2018optimizing,sanders2020Optimization} which use degree $\tilde{d}_\epsilon$ piecewise polynomial approximations.}
    \label{tab:MethodComparison}
\end{table*}

\paragraph*{Classical pre-computation.} The approximating polynomial $h(\cdot)$, which approximates $f(a\arcsin(\cdot))$, can be calculated using the Remez algorithm for minimax polynomials~\cite{remez1962general,fraser1965survey}, or via Taylor expansion. The requirement $|h(y)|_{\mathrm{max}}^{y \in [-1,1]} \leq  1$ ensures the QSVT circuit is unitary, regardless of the block-encoding to which it is applied, and may require multiplying the approximating polynomial by an approximate threshold function, to ensure that it is still less than 1 outside of the window $[-\sin(1), \sin(1)]$. We expect that this results in a modest increase in the degree of $h(y)$. Given the degree $d$ approximating polynomial $h(y)$, we can use efficient algorithms~\cite{chao2020finding, Haah2019product, dong2021efficient} to find the QSVT rotation angles.

Given a $\delta$-accurate approximating polynomial, the trace distance between $\ket{\Psi_{\tilde{f}}}$ and $\ket{\Psi_{f}}$ can be bounded as shown in Lemma~\ref{Lemma:ErrorBoundsRigorous}, using $\fillfrac{f}{N}$ \& $\fillfrac{\tilde{f}}{N}$. We can also use $\sum_x f(\bar{x}) \tilde{f}(\bar{x})$ to compute a tighter bound in practice (Appendix~\ref{App:TighterErrorAnalysis}). When $N$ is small, these terms can be evaluated directly, while when $N$ is large we approximate them by their continuous variants (e.g. $\fillfrac{f}{\infty}$).

\paragraph*{Comparison.} We contrast the scaling and features of our method with existing approaches that have rigorous error bounds in Table~\ref{tab:MethodComparison} (we do not compare against the heuristic matrix product state approach~\cite{garcia2021quantum,holmes2020efficient}, as it is unclear if it can achieve high accuracy).


\section{Applications}\label{Sec:Applications}

We apply our algorithm to prepare functions with important applications in quantum algorithms: Kaiser window and Gaussian functions. The Kaiser window function $W_\beta(x) =\frac{I_0(\beta\sqrt{1-x^2})}{I_0(\beta)}$ (where $I_0$ is the zeroth modified Bessel function of the first kind, see Appendix~\ref{App:Bessel}) can be used in quantum phase estimation (QPE)~\cite{berry2022quantifyingTDA,berry2025RapidInitial}. By preparing the QPE ancillas in this state, we can boost the success probability of QPE without (coherently) computing the median of multiple phase evaluations (see e.g.~\cite{Rall2021fastercoherent}). Gaussian states $f_\beta(x) = \exp(-\frac{\beta}{2}x^2)$ are widely used in quantum algorithms, e.g. in chemistry~\cite{kivlichan2017bounding, chan2022grid}, simulation of quantum field theories~\cite{jordan2012quantumfieldtheory,klco2021HierarchicalFieldTheory}, and finance~\cite{Chakrabarti2021thresholdquantum, stamatopoulos2020option}. In Appendix~\ref{App:GaussianComplexity} we prove the following theorem on the complexity of preparing Gaussian\footnote{Here $\beta$ should be thought of as $\frac{1}{\sigma^2}$, the inverse of the variance.} and Kaiser window states:

\begin{restatable}{theorem}{GaussianComplexity}\label{thm:GaussianComplexity}
	Let $f_\beta(x)$ be either $\exp(-\frac{\beta}{2}x^2)$ or $W_\beta(x)$. 
	If $\eps\in(0,\frac{1}{2})$ and $2^n\geq \sqrt{\beta}\geq 0$, then we can prepare the corresponding Gaussian / Kaiser window state on $n$ qubits up to $\eps$-precision with gate complexity
	\begin{align}\label{eq:KaiserComplexity}
		\bigO{ n  \sqrt[4]{\beta+1} \left(\beta+\log(1/\eps)\right)}.
	\end{align}
	For Gaussian states $f_\beta(x)=\exp(-\frac{\beta}{2}x^2)$ if $\beta \geq \log(1/\eps)$ this complexity can be further improved to 	
	\begin{align}\label{eq:GaussianComplexity}
		\bigO{n \log^{\frac{5}{4}}(1/\eps)}.
	\end{align}
\end{restatable}

\paragraph*{Kaiser window state.} In the Kaiser window state the parameter $\beta$ controls the trade-off between the central-band width and side-band height when viewed in the Fourier domain. In \Cref{App:Taylor} we show that $W_\beta(\arcsin(\bar{x}))$ can be approximated by a degree $\mathcal{O}\left(\beta + \ln\left(\delta^{-1}\right)\right)$ polynomial on the interval $x \in [-\sin(1), \sin(1)]$, utilizing the fact that $W_\beta(x)$ has a well behaved Taylor series.
To bound the filling-fraction, we show in \Cref{App:Filling} that $W_\beta(x) \geq 1-\beta x^2/2$. By integrating the lower bound for $\beta \geq 2$ we get that $\int_{-1}^1 W_\beta(x)^2 dx \geq  \sqrt{2/\beta}$. Hence $\fillfrac{W_b}{\infty} \geq \beta^{-1/4}$. This lower bound appears tight in practice, matching the true value with 85-90\% accuracy. 
Putting these bounds together with \Cref{thm:genericComplexity} gives the stated complexity in Eq.~(\ref{eq:KaiserComplexity}). For application in phase estimation, we can relate $\beta$ to the probability of failure $\eta$ as $\beta \sim \ln\left(\eta^{-1}\right)$, and $n$ to the precision $\epsilon_\phi$ of phase estimation as $n \sim \log\left(\epsilon_\phi^{-1} \ln\left(\eta^{-1}\right) \right)$~\cite{berry2022quantifyingTDA}. Hence, our method scales polylogarithmically in all parameters. We are not aware of any prior work discussing the complexity of preparing the Kaiser window state (which is also omitted from~\cite{berry2022quantifyingTDA}) or of resource estimates for implementing an amplitude oracle of the Bessel function, that could be used for the black-box or adiabatic state preparation methods.


\paragraph*{Gaussian state.} The proof of \Cref{thm:GaussianComplexity} for the Gaussian case is completely analogous to the Kaiser window case above. The bound can be tightened by observing that Gaussian functions take values close to zero for large $x$ values, and so one can assume without loss of generality that $\beta=\bigO{\log(1/\eps)}$, see \Cref{App:GaussianComplexity}.


\begin{table}[!ht]
    \centering
    \begin{tabular}{c|c|c}
        Method & \makecell{\# Ancilla \\ qubits} & \makecell{\# $T$ / Toffoli \\ gates} \\ \hline
        QSVT-based (This work) & $3$ & $48,000$ \\ \hline
        \makecell{Piecewise-polynomial~\cite{haner2018optimizing}} & $168$ & $120,000$ \\ \hline
        Linear interpolation~\cite{sanders2020Optimization} & $189$ & $ 24,000 $ \\ \hline
        Bespoke gaussian~\cite{poirier2021efficient} & $141$ & $45,000$ \\
    \end{tabular}
    \caption{Resources to prepare a quantum state representing $\exp(-\beta x^2)$ with $\beta=10$ and $x \in [-1,1]$, using $n=16$ qubits, with a trace distance $\epsilon \leq 10^{-6}$. We compare our QSVT-based method against the black-box state preparation approach~\cite{wang2021fast} with three different amplitude oracles.} 
    \label{tab:ResourceComparisonGaussian}
\end{table}

\paragraph*{Resource Estimates.} In Table~\ref{tab:ResourceComparisonGaussian} we compare the resources\footnote{While the cost of our method is most naturally expressed in $T$ gates, previous approaches are more naturally expressed in terms of Toffoli gates. One can convert 4 $T$ gates to a Toffoli using an ancilla qubit~\cite{jones2013LowOverheadToff}, or we can implement two $T$ gates from a CCZ state (equiv. Toffoli) using a $T$ state catalyst ancilla~\cite{gidney2019efficientmagicstate}.} to prepare a Gaussian state with our QSVT-based method, against the resources when using the LCU-based black-box state preparation approach~\cite{wang2021fast} with 3 different amplitude oracles; the piecewise-polynomial oracle~\cite{haner2018optimizing}, the linear interpolation oracle~\cite{sanders2020Optimization} (which can be viewed as maximally streamlining the piecewise polynomial approach) and a bespoke oracle for Gaussians~\cite{poirier2021efficient}\footnote{The estimates for the bespoke gaussian amplitude oracle are an optimistic lower bound, as the resource estimates available in~\cite{poirier2021efficient} consider $n=13$, and target a more peaked gaussian with $\beta=100$ (which results in a lower cost than $\beta=10$).}. We give a high level discussion of the costs here, and refer to Appendix~\ref{App:ResourceEst} for additional details. We expect that these methods will be more efficient than other bespoke methods for Gaussians such as: the Kitaev-Webb (KW) method~\cite{kitaev2008wavefunction}, and the repeat-until-success approach of~\cite{Rattew2021efficient}. The KW method is similar in spirit to Grover-Rudolph~\cite{grover2002creating}, and was shown to produce higher gate counts than exponentially scaling (in $n$) state preparation techniques for modest $n \leq 16$, due to the costly amplitude oracle required~\cite{bauer2021practical}. The approach of~\cite{Rattew2021efficient} has a circuit depth of $\mathcal{O}(n^2 \cdot Poly(\epsilon^{-1}))$, with a large constant prefactor.

\paragraph*{QSVT-based approach.} As discussed in Appendix.~\ref{App:ResourceEst}, the $T$ cost of our approach can be approximated by 
\begin{equation}
    (2R+1)d(n+1)(0.57 \log_2((2R+1)d(n+1)/\epsilon_s) + 8.83).
\end{equation} 
where $R$ is the number of rounds of amplitude amplification, $d$ is the degree of the approximation polynomial used, and $\epsilon_s$ is the rotation synthesis error (taken as $10^{-7}$ here). An even parity $d=20$ polynomial suffices to achieve a trace distance of around $5.7 \times 10^{-7}$. We calculate that $R=2$ in this example.

\paragraph*{Black-box approach.} We lower bound the cost by only counting non-Clifford gates due to the amplitude oracle. Each round of amplitude amplification (again $R=2$) calls the oracle and its inverse once, plus one final additional call for uncomputing garbage~\cite{sanders2019black,wang2021fast}. In addition to the ancilla costs of the amplitude oracle, the black-box method requires $2\log(n)-1 = 7$ ancilla qubits~\cite{wang2021fast}, and it requires 1 additional qubit for exact amplitude amplification. In all amplitude oracles we target an $L_\infty$ error $< 10^{-7}$. We remark that it is possible to halve the number of rounds of amplitude amplification (and thus the gate count) using the (more complex) prior-enhanced variant of the black-box approach in~\cite[Sec.IV.D.2]{Bagherimehrab2022GaussianFieldTheory}.

\paragraph*{Comparison.} Our approach reduces the ancilla count by over an order of magnitude, and yields a similar gate count to the amplitude oracle-based methods. We can further reduce the gate count of our method using a modest cost of $n$ additional qubits by eliminating the block-encoding rotation gates. One option is to use addition with an $n$-qubit phase gradient catalyst (cost $4n$ $T$ gates~\cite{Gidney2018halvingcostof}). Another option is to use the $n$ ancilla qubits to block-encode $x$ rather than $\sin(x)$, using the comparison test approach in~\cite{sanders2019black} (cost $2n-1$ Toffoli gates). By tailoring the block-encoding to minimize certain metrics (e.g. 2 qubit gates in NISQ, non-Clifford gates in the error corrected computations) we can make our method architecture specific.

\section{Extensions}\label{Sec:Extensions}

\paragraph*{Priors.} We can incorporate the use of improved priors in our method (cf.~\cite{bausch2020fast}). By applying $U_{\tilde{f}}$ to $\ket{+}^{\otimes n}$, we are choosing a uniform prior, leading to the $1/\fillfrac{\tilde{f}}{N}$ rounds of amplitude amplification. We can instead prepare $\ket{000} \mathcal{N}_p^{-1} \sum_x p(\bar{x}) \ket{x}$ and block-encode a polynomial approximation of $f(\bar{x})/p(\bar{x})$. We require $\mathcal{O}\left(\mathcal{N}_f^{-1}\mathcal{N}_p \left|f/p\right|_{\mathrm{max}} \right)$ rounds of amplitude amplification. If the prior distribution can be prepared with low cost, has a similar normalization to $f(\bar{x})$, and there exists a similar degree approximation of $f(\bar{x})/p(\bar{x})$ as there is for $f(\bar{x})$, this can reduce the resources required.

\paragraph*{Non-smooth functions.} We can extend our method to functions with a modest number of discontinuities, which are typically pathological for QSVT-based methods. Our application to state preparation enables us to circumvent this issue using two possible techniques. The first route uses a coherent inequality test to entangle the register with a flag qubit (such that the flag qubit is $\ket{0}/\ket{1}$ for $x$ to the left/right of the discontinuity). We control the rotations of the QSVT-ancilla on the flag, applying a different QSVT polynomial to each part of the register. For $k$ discontinuities, this piecewise extension requires $(k+n)$ ancilla qubits and $2kn$ Toffoli gates for the inequality comparison (and its uncomputation), and replaces the rotations of the ancilla by $k$ controlled rotations. 

The second route is more resource efficient when the number of discontinuities is small. As above, we perform a coherent inequality test to flag states to the right of the discontinuity point. We can view the ancilla as enlarging our domain, from an $n$-bit representation, to an $(n+1)$-bit representation, while maintaining the grid spacing. This opens a gap at the discontinuity point, such that the quantum state has no support on computational basis states in the vicinity of the discontinuity. We can then replace the original, discontinuous function by a continuous function that has the desired behaviour outside of the `gap' opened by the inequality test. Once the function has been applied, we can close the gap by uncomputing the inequality test. In exchange for the added complexity of block-encoding the function in wider range, we can replace the non-analytic function $\tilde{f}(\bar{x})$ with a continuously differentiable approximation, requiring a substantially lower degree polynomial.


\paragraph*{Fourier series.} Our method is naturally compatible with `Fourier-based quantum eigenvalue transformation'~\cite{silva2022fourier, dong2022ground} which provides a complementary approach for function approximation through Fourier series. In that approach, the block-encoding of $A$ is replaced by controlled time evolution $U(A) := \ket{0}\bra{0} \otimes I + \ket{1}\bra{1} \otimes e^{iAt}$, efficiently implementable for diagonal $A = \sum_x \bar{x} \ketbra{x}{x}$ using a controlled-phase-gradient operation~\cite{GidneyBlog2017}. Our methods are particularly appealing for functions with a compact Fourier series, such as spherical harmonic functions in chemistry.


\section{Outlook}
\paragraph*{Conclusion.} We have introduced a QSVT-based approach to preparing quantum states that represent continuous functions with polynomial approximations. By circumventing the coherent arithmetic instantiated amplitude oracle typically used, we can significantly reduce the number of ancilla qubits required. Our approach uses the same circuit template for all suitable functions, in contrast to the bespoke circuits typically developed as amplitude oracles. We have shown how to prepare Gaussian and Kaiser window functions with lower complexity than prior state-of-the-art approaches. We expect our technique to prove useful in a wide range of quantum algorithms, including those for chemistry and physics simulation, phase estimation, finance, and differential equation solving --- indeed it has already shown utility in these latter three applications~\cite{chen2023gibbs,stamatopoulos2023derivativeQSP, Li2023efficientquantum} and has been incorporated as an example in the open-source qsppack package~\cite{qsppackexamples2024}. 


\paragraph*{Multivariate functions.} A straightforward multivariate extension of our approach would use linear combinations /products of block-encodings~\cite{gilyen2018quantum} to implement a function $f(x,y)$ with a series expansion in powers of $x,y$. The expansion coefficients (which determine the final normalization of the block-encoding and thus the number of rounds of amplitude amplification) can be much smaller in the Fourier basis than in the polynomial basis. A potentially more efficient route to generate a multivariate function $f(\vec{x})$ may be to use the recently introduced multivariable-QSP~\cite{rossi2022multivariable}. Nevertheless, characterizing the functions that can be implemented via M-QSP is still an ongoing area of research~\cite{nemeth2023variants}. It is also unclear how to address the expected exponential decay of filling-fraction with dimension for multivariate functions.


\paragraph*{Acknowledgements.} We thank Fernando Brand\~ao for discussions and support throughout the project. A.G. acknowledges funding from the AWS Center for Quantum Computing. M.B. is supported by the EPSRC (Grant number EP/W032643/1).

\bibliography{Bib}


\appendix

\section{Signed integer representation}\label{App:SignedIntegers}
In this work we use the two's complement representation of signed integers. Using $n$ bits, we use the first (rightmost) $n-1$ bits to represent numbers from $0$ to $2^{n-1} - 1$. E.g. for $n-1=3$ we can represent the numbers from $0 = \ket{000}$ to $7 = \ket{111}$. The leftmost bit is used to control the sign as follows. If the $n$-th bit is in $\ket{0}$, the number represented by the rest of the binary string is unchanged. If the $n$-th bit is in $\ket{1}$, then we subtract $2^{n-1}$ from the number represented by the rest of the binary string. Hence, for $n=4$, $\ket{0000} = 0$, $\ket{0111} = 7$, $\ket{1000}=-8$, $\ket{1111}=-1$. Hence we can represent the $2^n$ integers between $-2^{n-1}$ and $2^{n-1} - 1$.

\section{Exact amplitude amplification}\label{App:AmpAmp}

In this appendix we describe exact amplitude amplification. This result is folklore, but we could not find a standard reference, especially one that treats the case when the amplitude is only approximately known, so we give a full treatment here. 

We utilize Chebyshev polynomials of the first kind defined as $T_{n}(x)=\cos(n\arccos(x))$, and their recurrence relation $T_{n+1}(x)=2x T_n(x)-T_{n-1}(x)$.

\begin{lemma}[Amplitude amplification]\label{lem:AmpAmp}
	Let $U$ be an $n$-qubit unitary, $\Pi$ an $n$-qubit projector, $\ket{\psi}$ an $n$-qubit (normalized) quantum state, and $a\geq 0$ such that
	\begin{align}\label{eq:stateprep}
		 \Pi U  \ket{\bar{0}} = a \ket{\psi},
	\end{align}  
	where $\ket{\bar{0}}$ denotes some $n$-qubit initial state.
	
	Let $W=U\left(2\ketbra{\bar{0}}{\bar{0}}-I\right)U^\dagger\left(2\Pi-I\right)$, then
	\begin{align}
		\Pi W^k U  \ket{\bar{0}} & = T_{2k+1}(a) \ket{\psi}, \quad \text{and}\label{eq:Cheby1}	\\	
		\bra{\bar{0}}U^\dagger(2\Pi-I) W^k U  \ket{\bar{0}} & = T_{2k+2}(a). \label{eq:Cheby2}
	\end{align} 
\end{lemma}
\begin{proof}
	\Cref{eq:Cheby1,eq:Cheby2} follow for $k=0$ from \eqref{eq:stateprep} using that $T_1(x)=x$ and $T_2(x)=2x^2-1$.
	
	We prove them for positive values of $k$ by induction:
	\begin{align*}
		\Pi W^{k+1} U  \ket{\bar{0}}&=\Pi U\left(2\ketbra{\bar{0}}{\bar{0}}-I\right)U^\dagger\left(2\Pi-I\right)W^{k} U  \ket{\bar{0}}\\&
		=\left(2a\ketbra{\psi}{\bar{0}}-\Pi U\right)U^\dagger\left(2\Pi-I\right)W^{k} U  \ket{\bar{0}}\\&
		=2a\ketbra{\psi}{\bar{0}}U^\dagger\left(2\Pi-I\right)W^{k} U \ket{\bar{0}}-\Pi W^{k} U \ket{\bar{0}}\\&		
		=\left(2a T_{2k+2}(a) - T_{2k+1}(a)\right)\ket{\psi}\\&	
		=T_{2k+3}(a)\ket{\psi},
	\end{align*}
	and
	\begin{align*}
		&\bra{\bar{0}}U^\dagger(2\Pi-I) W^{k+1} U  \ket{\bar{0}}\\&
		=2\bra{\bar{0}}U^\dagger\Pi W^{k+1} U  \ket{\bar{0}} - \bra{\bar{0}}U^\dagger W^{k+1} U  \ket{\bar{0}}\\&
		=2\bra{\bar{0}}U^\dagger\Pi\Pi W^{k+1} U  \ket{\bar{0}} - \bra{\bar{0}}U^\dagger W^{k+1} U  \ket{\bar{0}}\\&		
		=2a T_{2k+3}(a) - \bra{\bar{0}}U^\dagger W^{k+1} U  \ket{\bar{0}}\\&			
		=2a T_{2k+3}(a) - \bra{\bar{0}}U^\dagger U\left(2\ketbra{\bar{0}}{\bar{0}}-I\right)U^\dagger\left(2\Pi-I\right) W^{k} U  \ket{\bar{0}}\\&	
		=2a T_{2k+3}(a) - \bra{\bar{0}}U^\dagger\left(2\Pi-I\right) W^{k} U  \ket{\bar{0}}\\&					
		=2a T_{2k+3}(a) - T_{2k+2}(a)\\&		
		=T_{2k+4}(a).\qedhere
	\end{align*}	
\end{proof}

\begin{theorem}[Exact amplitude amplification]\label{thm:ExactAmpAmp}
	Suppose $U$, $\Pi$, $\ket{\psi}$, $\ket{\bar{0}}$, and $a$ are as in \Cref{lem:AmpAmp}. 
	Let $k:=\left\lceil\frac{\pi}{4\arcsin(a)}-\frac12\right\rceil$, and let $\theta:=\frac{\pi}{4k+2}$. Suppose that $R$ is a single-qubit unitary such that $\bra{0}R\ket{0}=\frac{\sin(\theta)}{a}$. Let us define $U':=R\otimes U$ and 
	\begin{align*}
		W':=U'\left(2\ketbra{0}{0}\otimes\ketbra{\bar{0}}{\bar{0}}-I\right)U'^\dagger\left(I-2\ketbra{0}{0}\otimes\Pi\right),
	\end{align*}
	then
	\begin{align}
		(W')^k U'  \ket{0}\ket{\bar{0}} = \ket{0}\ket{\psi}.\label{eq:ExCheby}
	\end{align} 
	Moreover, if $\tilde{a}\leq 2a< 2$ and $\tilde{U}$ is such that 	
	\begin{align}\label{eq:stateprepűtilde}
		\Pi \tilde{U}  \ket{\bar{0}} = \tilde{a} \ket{\tilde{\psi}},
	\end{align} 
	then
	\begin{align}
		\left(\ketbra{0}{0}\otimes\Pi\right) (\tilde{W}')^k \left(R\otimes\tilde{U}\right)  \ket{0}\ket{\bar{0}} = c\ket{0}\ket{\tilde{\psi}},\label{eq:ExChebyTilde}
	\end{align}
	for some $c\geq 1-(2k+1)(2k+2)|\tilde{a}-a|^2$, where $\tilde{W}'$ is defined analogously to $W'$ just $U'$ is replaced by $R\otimes\tilde{U}$.
\end{theorem}
\begin{proof}
	First note that 
	\begin{align*}
		\theta=\!\frac{\pi}{4\left\lceil\!\frac{\pi}{4\arcsin(a)\!}\!-\!\frac12\!\right\rceil\!+\!2}
		\!\leq\! \frac{\pi}{4\left(\!\frac{\pi}{4\arcsin(a)}\!-\!\frac12\!\right)\!+\!2}
		\!=\arcsin(a),
	\end{align*}
	and therefore $\bra{0}R\ket{0}=\frac{\sin(\theta)}{a}\leq 1$. Observe that
	$\left(\ketbra{0}{0}\!\otimes\!\Pi\right) U'  \ket{0}\ket{\bar{0}} \!=\! \left(\ket{0}\bra{0}R\ket{0}\right)\!\otimes\! \left(\Pi U  \ket{\bar{0}}\right)\!=\!\sin(\theta)\ket{0}\ket{\psi}$.
	Applying \Cref{lem:AmpAmp} with $U'$, $\Pi':=\ketbra{0}{0}\otimes\Pi$, $\ket{\psi'}:=\ket{0}\ket{\psi}$, $\ket{\bar{0}'}:=\ket{0}\ket{\bar{0}}$, and $a':=\sin(\theta)$ we get that 
	\begin{align*}
		\Pi' (-W')^k U'\ket{\bar{0}'} = T_{2k+1}(\sin(\theta)) \ket{\psi'},
	\end{align*}
	thus
	\begin{align*}
		\Pi' (W')^k U'\ket{\bar{0}'} &= (-1)^k T_{2k+1}(\sin(\theta)) \ket{\psi'}
		=\ket{\psi'},
	\end{align*}
	where the last equality holds because
	\begin{align*}
		(-1)^k T_{2k+1}(\sin(\theta)) 
		&=(-1)^k \cos((2k+1)\arccos(\sin(\theta))) \\&
		=(-1)^k \cos((2k+1)(\pi/2-\theta)) \\&
		=(-1)^k \cos(k\pi) = 1.
	\end{align*}

	Similarly, by \Cref{lem:AmpAmp} we get that
	\begin{align*}
		\Pi' (\tilde{W}')^k \left(R\otimes\tilde{U}\right)\ket{\bar{0}'} &= (-1)^k T_{2k+1}\left(\sin(\theta)\frac{\tilde{a}}{a}\right) \ket{0}\ket{\tilde{\psi}}.
	\end{align*}
	As we have seen $(-1)^k T_{2k+1}\left(y\right)$ takes value $1$ at $y=\sin(\theta)$, which also implies that its derivative is $0$ there since $|T_{2k+1}\left(y\right)|\leq 1$ for all $y\in [-1,1]$ and $\sin(\theta)<1$ (as $a<1$). By Taylor's theorem we have that $(-1)^k T_{2k+1}\left(\sin(\theta) + \xi\right)\geq 1 - \frac{M_2}{2}\xi^2$, where $M_2$ is the maximal absolute value of the second derivative of $T_{2k+1}\left(y\right)$ at any point between $\sin(\theta)$ and $\sin(\theta+\xi)$. 
	Observe that $|\sin(\theta)\frac{\tilde{a}}{a}-\sin(\theta)|=\frac{\sin(\theta)}{a}|\tilde{a}-a|\leq |\tilde{a}-a|$ so in Taylor's theorem we can bound $|\xi|\leq |\tilde{a}-a|$.
	
	If $\tilde{a}\leq 2a$ then $\max\{\sin(\theta),\sin(\theta)\frac{\tilde{a}}{a}\}\leq 2 \sin(\theta)$, so it suffices to bound the magnitude of the second derivative $|T_{2k+1}''(y)|$ for $y\in [-2\sin(\theta),2\sin(\theta)]$.
	If $a\in[ \frac12,1)$, then $k=1$ and $|T_{3}''\left(y\right)|=|24 y|\leq 2(2k+1)(2k+2)$ so $M_2\leq 2(2k+1)(2k+2)$.
	If $a\in[\sin(\pi/10), \frac12)$, then $k=2$ and $|T_{5}''\left(y\right)|=|320 y^3-120y|\leq (2k+1)(2k+2)$ for $y\in[-2\sin(\pi/10),2\sin(\pi/10)]$ so $M_2\leq (2k+1)(2k+2)$.
	Finally, for $a < \sin(\pi/10)$ we have $k\geq 3$ and $2\sin(\theta)\leq 2\sin(\pi/14) < 0.45$. 
	Considering $\alpha:=n \arccos(y)$ and $y\in[-1,1]$ we have $|T_{n}''\left(y\right)|=n\left|\frac{n\cos(\alpha)\sqrt{1-y^2}-y\sin(\alpha)}{\left(1-y^2\right)^{\frac32}}\right|
	\leq \frac{n(n+1)}{{\left(1-y^2\right)^{\frac32}}}$ 
	which is $\leq 2 n(n+1)$ for $y\in[-\frac12,\frac12]$.
	This completes the case separation and proves that $M_2/2\leq (2k+1)(2k+2)$ implying that $c\geq 1-(2k+1)(2k+2)|\tilde{a}-a|^2$.
\end{proof}

\subsection{Working with approximately known amplitudes}\label{AppSub:ApproxAmps}

We discuss how best to amplify the state in cases where we do not know the exact value of $\fillfrac{\tilde{f}}{N}$. This may arise because the value $n$ is so large that it would be too costly to classically compute the filling fraction. If we have a lower bound for $\fillfrac{\tilde{f}}{N}$, then we can simply apply fixed-point amplitude amplification, using QSVT~\cite{gilyen2018quantum}. This also only uses a single additional ancilla qubit\footnote{For the implementation of the generalized Toffoli required for the reflection around the all-$0$ initial state we might need an additional second ancilla qubit.\label{foot:ToffoliAncilla}} and increases the success probability to $\geq (1-\zeta)$ at the cost of a multiplicative overhead of $\bigO{\log\left( \zeta^{-1} \right)}$.

If $n$ is sufficiently large, it is possible to approximate the value of $\fillfrac{\tilde{f}}{N}$ by its continuous counterpart $\fillfrac{\tilde{f}}{\infty}$ or $\fillfrac{f}{\infty}$, c.f. \Cref{Apx:genDiscBounds}, which is efficient to evaluate for many functions. Assuming that $\left|\fillfrac{\tilde{f}}{\infty} - \fillfrac{\tilde{f}}{N} \right| \leq \delta \leq \fillfrac{\tilde{f}}{\infty}$, we can apply \Cref{thm:ExactAmpAmp} for bounding the error in the resulting amplitude by \[\bigO{\bigg(\frac{\delta}{\fillfrac{\tilde{f}}{\infty}}\bigg)^{\!2}}.\]
As the approximation error $\delta$ decreases exponentially with the number of qubits $n$ used for discretizing the function, we expect this error to be small.

\subsection{General discretization error bounds}\label{Apx:genDiscBounds}

Here we recall some standard results on Riemann sums. The first result considers our default discretization method but has a looser bound, while the second improves upon it but requires a slightly different placing of the discrete points.

\begin{lemma}[see \cite{grinshpan2009AnalysisNotes}]
	Suppose that $f\colon [a,b]\rightarrow \mathbb{R}$ is continuously differentiable. 
	Let $\bar{x} = \left( (b-a)x/N + a\right)$, then
	\begin{align*}
		\left|\frac{b-a}{N}\sum_{x=0}^{N-1} f(\bar{x})-\int_a^b f(x) dx\right|\leq \frac{(b-a)^2}{2N}|f'(x)|_{\mathrm{max}}^{x \in [a,b]}.
	\end{align*}
\end{lemma}

\begin{lemma}[see \cite{grinshpan2009AnalysisNotes}]
	Suppose that $f\colon [a,b]\rightarrow \mathbb{R}$ is twice continuously differentiable. 
	Let $\bar{x} = \left( (b-a)(x+\frac{1}{2})/N + a\right)$, then
	\begin{align*}
		\left|\frac{b-a}{N}\sum_{x=0}^{N-1} f(\bar{x})-\int_a^b f(x) dx\right|\leq \frac{(b-a)^3}{24N^2}|f''(x)|_{\mathrm{max}}^{x \in [a,b]}.
	\end{align*}
\end{lemma}


\section{Proof of Theorem 1}\label{App:ErrorAnalysis}

In this Appendix we prove \Cref{thm:genericComplexity}, which bounds the gate complexity of our method. We present a slightly more formal version of \Cref{thm:genericComplexity}, which makes use of the following definitions:

\begin{restatable}{definition}{DefOne}\label{Def:Psi}
    \begin{equation*}
    \ket{\Psi_f} := \frac{1}{\norm{f}} \sum_{x=-\frac{N}{2}}^{\frac{N}{2}-1} f\left( \bar{x} \right) \ket{x},
\end{equation*}
where $f\colon[-a, a]\rightarrow \mathbb{R}$ has definite-parity, $N=2^n$, $\bar{x} :=\left( 2ax/N\right)$, and $\mathcal{N}_f := \sqrt{\sum |f(\cdot)|^2}$. We use a two's complement representation of signed integers (see Appendix~\ref{App:SignedIntegers}).
\end{restatable}

\begin{restatable}{definition}{DefTwo}\label{Def:FillFrac}
   For a function $p(y)$ in the range $y \in [-a,a]$ define the `discretized L2-norm filling-fraction'
\begin{equation}
    \fillfrac{p}{N} = \frac{\norm{p}}{\sqrt{N} |p(y)|_{\mathrm{max}}^{y \in [-a,a]}}
\end{equation}
which approximates the continuous quantity $\fillfrac{p}{\infty} := \sqrt{ \frac{\int_{-a}^a |p(y)|^2 dy}{2a \left(|p(y)|_{\mathrm{max}}^{y \in [-a,a]} \right)^2} }$. 
\end{restatable}

We now restate and prove Theorem~\ref{thm:genericComplexity} (as Theorem~\ref{thm:genericComplexityFormal}).

\begin{restatable}{theorem}{GeneralFormal}\label{thm:genericComplexityFormal}
For a definite-parity function $f(\cdot)$ on the interval $[-a,a]$, define $\ket{\Psi_f}$ as in Definition~\ref{Def:Psi}. We are given a degree $d$ definite-parity polynomial $h(y)$, obeying $|h(y)|_{\mathrm{max}}^{y \in [-1,1]} \leq  1$, which approximates $f(\cdot)$ as 
\begin{equation}
    \left|\tilde{f}(y) - \frac{f(a y)}{\mx{f(ay)}^{y \in [-1,1]}}\right|_{\mathrm{max}}^{y \in [-1,1]}  \leq \frac{\epsilon~\cdot~\mathrm{Min}\left(\fillfrac{f}{N},  \fillfrac{\tilde{f}}{N} \right)}{3} 
\end{equation}
where $\tilde{f}(y) := h(\sin(y/a))$. Then we can prepare a quantum state $\ket{\Psi_{\tilde{f}}}$ that is no more than $\epsilon$-far from $\ket{\Psi_f}$ in trace distance using a quantum circuit requiring $\mathcal{O}\left( \frac{n d}{\fillfrac{\tilde{f}}{N}}  \right)$ gates and at most 3 ancilla qubits.
\end{restatable}
\begin{proof}
Using the results of Lemma~\ref{Lemma:BlockEncodeSin} we can implement a $(1,1,0)$ block-encoding $U_{\sin}$ of the $n$ qubit operator $\sum_{x = -\frac{N}{2}}^{\frac{N}{2}-1} \sin \left( \frac{2x}{N} \right) \ket{x}\bra{x} $, using $\mathcal{O}(n)$ elementary single- and two-qubit gates. By the results of Lemma~\ref{Lemma:QSVT} we can implement a $(1,2,0)$ block-encoding $U_{\tilde{f}}$ of the $n$ qubit operator 
\begin{align}
    &\sum_{x = -\frac{N}{2}}^{\frac{N}{2}-1} h\left(\sin \left( \frac{2x}{N} \right)\right) \ket{x}\bra{x} \\
    = &\sum_{x = -\frac{N}{2}}^{\frac{N}{2}-1} \tilde{f}(\bar{x}) \ket{x}\bra{x}
\end{align}
using $\mathcal{O}(d)$ calls to $U_{\sin}$ and $U_{\sin}^\dag$, and $\mathcal{O}(d)$ additional elementary gates. Lemma~\ref{Lemma:QSVT} is applicable by the assumption that $|h(y)|_{\mathrm{max}}^{y \in [-1,1]} \leq  1$. This property further guarantees that $|h(\sin(y/a))|_{\mathrm{max}}^{y \in [-1,1]} \leq  1$.

Applying $U_{\tilde{f}}$ to the state $\ket{00} \frac{1}{\sqrt{N}} \sum_{x = -\frac{N}{2}}^{\frac{N}{2}-1} \ket{x}$ outputs
\begin{equation}
    \ket{00}\left(\frac{1}{\sqrt{N}} \sum_{x = -\frac{N}{2}}^{\frac{N}{2}-1} \tilde{f}(\bar{x}) \ket{x}\right) + \ket{\perp}
\end{equation}
where $\ket{\perp}$ is an $(n+2)$ qubit state orthogonal to $\ket{00}$. Measuring the first two ancilla qubits in $\ket{00}$ produces the state $\ket{\Psi_{\tilde{f}}} = \frac{1}{\norm{\tilde{f}}} \sum_{x = -\frac{N}{2}}^{\frac{N}{2}-1} \tilde{f}(\bar{x}) \ket{x}$ with success probability 
\begin{equation}
    \frac{\norm{\tilde{f}}^2}{N} = \left(|\tilde{f}(y)|_{\mathrm{max}}^{y \in [-1,1]} \fillfrac{\tilde{f}}{N} \right)^2.
\end{equation}

Using the bound
\begin{equation}
    \left|\tilde{f}(y) - \frac{f(a y)}{\mx{f(ay)}^{y \in [-1,1]}}\right|_{\mathrm{max}}^{y \in [-1,1]}  \leq \frac{\epsilon~\cdot~\mathrm{Min}\left(\fillfrac{f}{N},  \fillfrac{\tilde{f}}{N} \right)}{3} 
\end{equation}
ensures that 
\begin{align}
    |\tilde{f}(y)|_{\mathrm{max}}^{y \in [-1,1]} &\geq 1 - \frac{\epsilon~\cdot~\mathrm{Min}\left(\fillfrac{f}{N},  \fillfrac{\tilde{f}}{N} \right)}{3} \\
    &\geq \frac{2}{3}
\end{align}
where we have used that $\epsilon, \mathrm{Min}\left(\fillfrac{f}{N},  \fillfrac{\tilde{f}}{N} \right) \leq 1$.

Hence the success probability is lower bounded by $\frac{4}{9}\left(\fillfrac{\tilde{f}}{N} \right)^2 $. Using the results of exact amplitude amplification from Theorem~\ref{thm:ExactAmpAmp}, the success probability can be boosted to unity using a quantum circuit that makes $\mathcal{O}\left(1 / \fillfrac{\tilde{f}}{N}\right)$ calls to $U_{\tilde{f}}$, $U_{\tilde{f}}^\dag$, and requires $\mathcal{O}\left(n / \fillfrac{\tilde{f}}{N}\right)$ additional elementary gates to implement the reflection operators. The circuit requires at most one additional ancilla qubit.

The circuit thus uses $\mathcal{O}\left( \frac{n d}{\fillfrac{\tilde{f}}{N}}  \right)$ gates and at most 3 ancilla qubits to prepare the state $\ket{\Psi_{\tilde{f}}}$ with probability 1. By the results of Lemma~\ref{Lemma:ErrorBoundsRigorous}, this state is no more than $\epsilon$-far in trace distance from $\ket{\Psi_{f}}$.
\end{proof}

\subsection{Lemmas for proving Theorem 1}

\begin{lemma}\label{Lemma:BlockEncodeSin}
    There exists a quantum circuit $U_{\mathrm{sin}}$ that implements a $(1,1,0)$-block-encoding of the $n$ qubit operator $\sum_{x = -\frac{N}{2}}^{\frac{N}{2}-1} \sin \left( \frac{2x}{N} \right) \ket{x}\bra{x} $. The circuit $U_{\mathrm{sin}}$ uses $\mathcal{O}(n)$ elementary single- and two-qubit gates.
\end{lemma}
\begin{proof}
Define $R_z(\theta) = \mathrm{Diag}(1, e^{i\theta})$.
First, observe that the following two-qubit circuit with $y \in \{0,1\}$
\[
\Qcircuit @C=.5em @R=0.2em @!R {
	\lstick{\ket{0}} & \push{\rule{.1em}{0em}} & \gate{H} & \ctrl{1} & \qw & \qw & \ctrl{1} & \qw & \gate{R_z(-\theta)} & \gate{H} & \gate{Y} & \qw \\
	\lstick{\ket{y}} & \push{\rule{.1em}{0em}} & \qw & \targ & \qw & \gate{R_z(\theta)} & \targ & \qw & \qw & \qw & \qw & \qw \\
}
\]
transforms 
\begin{equation}
    \ket{0}\ket{y} \rightarrow \left(\sin(\theta \cdot y)\ket{0} + i\cos(\theta \cdot y)\ket{1}\right)\ket{y}.
\end{equation}

Second, consider the following sequence of $R_z$ rotations acting on $n$ qubits~\cite{GidneyBlog2017}:
\begin{align}\label{Eq:RzCircuit}
    R_z\left(-2^{0}\right) \ket{x_{n-1}} \left( \bigotimes_{j=n-2}^{j=0} R_z(2^{j-(n-1)}) \ket{x_j} \right) \\
    = e^{i \left(-x_{n-1} + \sum_{j=0}^{n-2} 2^{j} 2^{-(n-1)} x_j \right) } \ket{x_{n-1}} ... \ket{x_0}.
\end{align}
Using the signed integer representation in Appendix~\ref{App:SignedIntegers}, we express the $n$ bit integer $x$ as $x = - 2^{n-1}x_{n-1} + \sum_{j=0}^{n-2} 2^j x_j$. Hence, the above sequence of $R_z$ rotations implements the transformation
\begin{equation}
    \ket{x} = \ket{x_{n-1}} ... \ket{x_0} \rightarrow e^{i x / 2^{n-1}} \ket{x}.
\end{equation}

Combining these two circuits as $U_{\mathrm{sin}}$
\[
\Qcircuit @C=.5em @R=0.2em @!R {
	\lstick{\ket{0}} & \push{\rule{.1em}{0em}} & \gate{H} & \ctrl{4} & \qw & \qw & \ctrl{4} & \qw & \gate{R_z(\phi)} & \gate{H} & \gate{Y} & \qw \\
	\lstick{\ket{x_0}} & \push{\rule{.1em}{0em}} & \qw & \targ & \qw & \gate{R_z\left(2^{1-n}\right)} & \targ & \qw & \qw & \qw & \qw & \qw \\
	\lstick{\ket{x_1}} & \push{\rule{.1em}{0em}} & \qw & \targ & \qw & \gate{R_z\left(2^{2-n}\right)} & \targ & \qw & \qw & \qw & \qw & \qw \\
	\lstick{\vdots} & \push{\rule{.1em}{0em}} & \qw & \targ & \qw & \gate{\vdots} & \targ & \qw & \qw & \qw & \qw & \qw \\
	\lstick{\ket{x_{n-1}}} & \push{\rule{.1em}{0em}} & \qw & \targ & \qw & \gate{R_z\left(-2^{0}\right)} & \targ & \qw & \qw & \qw & \qw &\qw
}
\]
with $\phi = 1 - \sum_{j=0}^{n-2} 2^{j} 2^{-(n-1)}$, yields the transformation
\begin{equation}
    U_{\mathrm{sin}} \ket{0}\ket{x} \rightarrow \left(\sin\left(\frac{2x}{2^n}\right)\ket{0} + i\cos\left(\frac{2x}{2^n}\right)\ket{1}\right)\ket{x}.
\end{equation}
We see that
\begin{align}
    \left(\bra{0} \otimes I_n\right) U_{\mathrm{sin}} \left(\ket{0} \otimes I_n\right) \\
    = \left(\bra{0} \otimes I_n\right) U_{\mathrm{sin}} \left(\ket{0} \otimes \sum_x \ket{x}\bra{x}\right) \\
    = \sum_x \sin\left(\frac{2x}{N}\right) \ket{x}\bra{x}
\end{align}
where we have used that $N = 2^n$. Hence, $U_{\mathrm{sin}}$ is a $(1,1,0)$ block-encoding of $\sum_{x = -\frac{N}{2}}^{\frac{N}{2}-1} \sin \left( \frac{2x}{N} \right) \ket{x}\bra{x} $. The circuit $U_{\mathrm{sin}}$ uses $\mathcal{O}(n)$ elementary single- and two-qubit gates.
\end{proof}

\begin{lemma}\label{Lemma:QSVT}
    Given a degree $d$ polynomial $h(\cdot)$ of definite parity with the constraint $\max_{y \in [-1,1]} |h(y)| \leq 1$, and a $(1,m,0)$ block-encoding $U_A$ of a Hermitian operator $A$, there exists a quantum circuit $U_{h}$ that implements a $(1,m+1,0)$ block-encoding of the operator $h(A)$. The circuit $U_h$ makes $d/2$ calls to $U_A$, $d/2$ calls to $U_A^\dag$, and uses $\mathcal{O}(md)$ additional elementary single- and two-qubit gates.
\end{lemma}
\begin{proof}
    This follows directly from the results of \cite[Lemma 18]{gilyen2018quantum}, using quantum singular value transformation (QSVT) applied to the block-encoding $U_{A}$.
\end{proof}

\begin{lemma}\label{Lemma:ErrorBoundsRigorous}
    For a definite-parity function $f : [-a,a] \rightarrow \mathbb{R}$ define the $n$-qubit state $\ket{\Psi_f}$ as in Def.~\ref{Def:Psi}, and $\fillfrac{f}{N}$ as in Def.~\ref{Def:FillFrac}. Given a definite parity function $\tilde{f}(\cdot)$ such that $|\tilde{f}(y)|_{\mathrm{max}}^{y \in [-1,1]} \leq 1$ and $\left| \tilde{f}(y) - \frac{f(a y)}{\mx{f(ay)}^{y \in [-1,1]}}\right|_{\mathrm{max}}^{y \in [-1,1]}  \leq \frac{1}{3}\epsilon \cdot \mathrm{Min}\left(\fillfrac{f}{N}, \fillfrac{\tilde{f}}{N} \right) $, then the corresponding quantum states $\ket{\Psi_{f}}$ and $\ket{\Psi_{\tilde{f}}}$ are at worst $\epsilon$ far-apart in trace distance.
\end{lemma}
\begin{proof}
    First renormalize the polynomials $f(\cdot)$ and $\tilde{f}(\cdot)$ to ensure their maximum absolute values correspond to $-1$ or $1$.
    Define $\fu(y) := \frac{f(a y)}{\mx{f(ay)}^{y \in [-1,1]}}$, such that $\left| \tilde{f}(y) - \fu(y) \right|_{\mathrm{max}}^{y \in [-1,1]} \leq \delta'$ for a chosen $\delta'$. It is given that $\left| \tilde{f}(y) \right|_{\mathrm{max}}^{y \in [-1,1]} \leq 1$. To account for the case where $\tilde{f}(y)$ is subnormalized, let $\left| \tilde{f}(y) \right|_{\mathrm{max}}^{y \in [-1,1]} = 1 - \kappa \geq 1-\delta'$. Then
    \begin{align}
        &\left| \frac{\tilde{f}(y)}{1-\kappa} - \fu(y) \right|_{\mathrm{max}}^{y \in [-1,1]} \\
        &\leq \frac{1}{1-\kappa}\left| \tilde{f}(y) - \fu(y) \right|_{\mathrm{max}}^{y \in [-1,1]} + \frac{\kappa}{1-\kappa} \left| \fu(y) \right|_{\mathrm{max}}^{y \in [-1,1]} \\
        &\leq \frac{\delta' + \kappa}{1-\kappa} \leq \frac{2\delta'}{1-\delta'} := \delta.
    \end{align}
    Accordingly, define $\fut(y) := \frac{\tilde{f}(y)}{1-\kappa}$, ensuring that
    \begin{align}
        \left| \fu(y) \right|_{\mathrm{max}}^{y \in [-1,1]} &= 1 \\
        \left| \fut(y) \right|_{\mathrm{max}}^{y \in [-1,1]} &= 1 \\
        \left| \fut(y) - \fu(y) \right|_{\mathrm{max}}^{y \in [-1,1]} &\leq \delta
    \end{align}
    
    Second, observe that this normalization does not change the definition of the corresponding quantum states. This is because for a polynomial $p(x)$ normalized by a constant $c$
    \begin{align}
        \ket{\Psi_{\frac{p}{c}}} &:= \frac{1}{\norm{\frac{p}{c}}} \sum_{x = -\frac{N}{2}}^{\frac{N}{2}-1} \frac{p(2ax/N)}{c} \ket{x} \\
        &= \frac{c}{\norm{p}} \sum_{x = -\frac{N}{2}}^{\frac{N}{2}-1} \frac{p(2ax/N)}{c} \ket{x} \\
        &= \ket{\Psi_{p}}
    \end{align}
    Hence, renormalizing the functions as above does not change their trace distance.

    We can thus bound $\mathcal{D} \left( \ket{\Psi_{f}}, \ket{\Psi_{\tilde{f}}} \right)$ by exploiting its equality with $\mathcal{D} \left( \ket{\Psi_{\fu}}, \ket{\Psi_{\fut}} \right)$.

    We first bound $| \braket{\Psi_{\fu}}{\Psi_{\fut}} |^2$. The relation $\left| \fut(y) - \fu(y) \right|_{\mathrm{max}}^{y \in [-1,1]} \leq \delta$ implies $\fu(y) \fut(y) \geq \frac{1}{2}\left(\fu(y)^2 + \fut(y)^2 - \delta^2 \right)$. Thus,
    \begin{align}
    |\braket{\Psi_{\fu}}{\Psi_{\fut}} |^2 &= \left| \frac{1}{\norm{\fu} \norm{\fut} } \sum_x \fu(\bar{x}) \fut(\bar{x}) \right|^2 \\
    &\geq \left| \frac{1}{2 \norm{\fu} \norm{\fut} } \sum_x |\fu(\bar{x})|^2 + |\fut(\bar{x})|^2 - \delta^2 \right|^2 \\
    &= \frac{1}{4}\left| \frac{\norm{\fu}}{\norm{\fut}} + \frac{\norm{\fut}}{\norm{\fu}} - \frac{N \delta^2}{\norm{\fut} \norm{\fu}} \right|^2
\end{align}

Expanding out the square gives
\begin{align}
    \frac{1}{4}\left( \frac{\norm{\fu}^2}{\norm{\fut}^2} + \frac{\norm{\fut}^2}{\norm{\fu}^2} + 2 - \frac{2N \delta^2}{\norm{\fut} \norm{\fu}} \left( \frac{\norm{\fu}}{\norm{\fut}} + \frac{\norm{\fut}}{\norm{\fu}} \right) + \left( \frac{N \delta^2}{\norm{\fut} \norm{\fu}} \right)^2 \right)
\end{align}
Let $\norm{\fu}= A$ and $\norm{\fut} = B$ in the first two terms. We can use
\begin{equation}
    \frac{A^2}{B^2} + \frac{B^2}{A^2} \geq 2
\end{equation}
(as $(A^2 - B^2)^2 \geq 0$) to simply the above expression to
\begin{align}
    \geq \frac{1}{4}\left( 4 - \frac{2N \delta^2}{\norm{\fut} \norm{\fu}} \left( \frac{\norm{\fu}}{\norm{\fut}} + \frac{\norm{\fut}}{\norm{\fu}} \right) + \left( \frac{N \delta^2}{\norm{\fut} \norm{\fu}} \right)^2 \right).
\end{align}
We can drop the final term, as it strictly increases the value of the expression
\begin{align}
    &\geq 1 - \frac{N \delta^2}{2\norm{\fut} \norm{\fu}} \left( \frac{\norm{\fu}}{\norm{\fut}} + \frac{\norm{\fut}}{\norm{\fu}} \right) \\ 
    & = 1 - \frac{N \delta^2}{2} \left( \frac{\norm{\fu}^2 + \norm{\fut}^2}{\norm{\fut}^2 \norm{\fu}^2} \right) \label{AppEq:B12}
\end{align}
We now define $\alpha = \mathrm{Max}(\norm{\fut}, \norm{\fu})$, $\beta = \mathrm{Min}(\norm{\fut}, \norm{\fu})$, such that $\alpha \geq \beta$ (thus $\beta^2/\alpha^2 \leq 1)$ . Then
\begin{align}
    \frac{\alpha^2 + \beta^2}{\alpha^2 \beta^2} &= \frac{\alpha^2 \left(1 + \frac{\beta^2}{\alpha^2} \right)}{\alpha^2 \beta^2} \leq \frac{2}{\beta^2}.
\end{align}
Eq.~(\ref{AppEq:B12}) then becomes $\geq 1 - \frac{N \delta^2}{\beta^2}$, with $\beta = \mathrm{Min}(\norm{\fut}, \norm{\fu})$. We now examine the value $N / \beta^2$. Without loss of generality, choose $\beta = \norm{\fu}$ here. Then we have
\begin{align}
    \frac{N}{\norm{\fu}^2} =\frac{N \mx{f}^2}{\norm{f}^2} = \left( \fillfrac{f}{N} \right)^{-2}
\end{align}
using the definition of the discretized L2-filling fraction.
Similarly,
\begin{align}
    \frac{N}{\norm{\fut}^2} =\frac{N \mx{\fut}^2}{\norm{\fut}^2 } = \frac{N \mx{\tilde{f}}^2}{\norm{\tilde{f}}^2 } = \left( \fillfrac{\tilde{f}}{N} \right)^{-2}.
\end{align}
Hence,
\begin{align}
    |\braket{\Psi_{\fu}}{\Psi_{\fut}} |^2 \geq 1 - \left( \frac{\delta}{\mathrm{Min}\left(\fillfrac{f}{N}, \fillfrac{\tilde{f}}{N} \right)} \right)^2 
\end{align}
As a result,
\begin{align}
    \mathcal{D}(\ket{\Psi_{\fut}}, \ket{\Psi_{\fu}}) &\leq \frac{\delta}{\mathrm{Min}\left(\fillfrac{f}{N},  \fillfrac{\tilde{f}}{N} \right)}.
\end{align}
The equivalence between $\mathcal{D} \left( \ket{\Psi_{f}}, \ket{\Psi_{\tilde{f}}} \right)$ and $\mathcal{D} \left( \ket{\Psi_{\fu}}, \ket{\Psi_{\fut}} \right)$ yields
\begin{align}
    \mathcal{D}(\ket{\Psi_{\tilde{f}}}, \ket{\Psi_{f}}) &\leq \frac{\delta}{\mathrm{Min}\left(\fillfrac{f}{N},  \fillfrac{\tilde{f}}{N} \right)} \\
    &= \frac{2\delta'}{(1-\delta')\mathrm{Min}\left(\fillfrac{f}{N},  \fillfrac{\tilde{f}}{N} \right)}.
\end{align}
Observe that $\fillfrac{f}{N},  \fillfrac{\tilde{f}}{N} \leq 1$ and choose $\delta' \leq 1/3$. Then
\begin{equation}
    \delta' := \frac{1}{3}\epsilon \cdot \mathrm{Min}\left(\fillfrac{f}{N}, \fillfrac{\tilde{f}}{N} \right)
\end{equation}
ensures
\begin{equation}
    \mathcal{D}(\ket{\Psi_{\tilde{f}}}, \ket{\Psi_{f}}) \leq \epsilon.
\end{equation}

\end{proof}


\section{Tighter error analysis}\label{App:TighterErrorAnalysis}

The error bound in Lemma~\ref{Lemma:ErrorBoundsRigorous} is an overly pessimistic error bound, as it assumes that the error in the function approximation is the same at every point. For approximation methods such a Taylor series, the maximum error can be considerably larger than the average error. As a result, we can directly calculate the trace distance between the states
\begin{align}
    D(\ket{\Psi_{\tilde{f}}}, \ket{\Psi_f}) &= \sqrt{1 - |\braket{\Psi_f}{\Psi_{\tilde{f}}}|^2} \\ \nonumber
    &= \sqrt{1 - \bigg{|} \sum_x \frac{f(\bar{x}) \tilde{f}(\bar{x})}{\mathcal{N}_f \cdot \mathcal{N}_{\tilde{f}} }} \bigg{|}^2.
\end{align}
For a sufficiently large number of discretization points
\begin{equation}
    \int_a^b y(\bar{x}) d\bar{x} \approx \frac{(b-a)}{N} \sum_{x=0}^{N-1} y(\bar{x}),
\end{equation}
as shown in \Cref{Apx:genDiscBounds}, which lets us approximate the trace distance between the states by
\begin{equation}
    \sqrt{1 - \bigg{|} \frac{\int_a^b f(\bar{x}) \tilde{f}(\bar{x}) d\bar{x}}{\twonm{f} \cdot  \twonm{\tilde{f}}}} \bigg{|}^2.
\end{equation}


\section{Modified Bessel functions}\label{App:Bessel}
In this appendix we list some properties of modified Bessel functions that we use later for analyzing Kaiser Windows.
First let us recall \cite[Eq. (9.6.12)]{abramowitz1974HandbookMathFuns} the Taylor series of $I_0(z)$:
\begin{align}\label{eq:I0Taylor}
	I_0(z)=\sum_{k=0}^\infty \frac{(z^2/4)^k}{(k!)^2}.
\end{align}
We will also use the following integral representations \cite[Eqs. (9.6.18-9.6.19)]{abramowitz1974HandbookMathFuns}:
\begin{align}
	I_n(z)&=\frac{1}{\pi}\int_{0}^\pi \exp(z \cos(\theta))\cos(n\theta)d\theta\label{eq:I0CosIntegral}\\
	&=\frac{(\frac{z}{2})^n}{\sqrt{\pi}\Gamma(n+\frac{1}{2})}\int_{0}^\pi \exp(z \cos(\theta))\sin^{2n}(\theta)d\theta.
	\label{eq:I0SinIntegral}
\end{align}

\section{Taylor series truncation bounds}\label{App:Taylor}
Let us introduce some notation that we use throughout this appendix. For a function $h\colon \mathbb{R}\rightarrow \mathbb{C}$ that is analytic in a neighborhood of $0$ so that $h(y)=\sum_{k=0}^\infty b_k y^k$ we denote by $\vertiii{h}_1:=\sum_{k=0}^\infty|b_k|$ the sum of the absolute values of the Taylor coefficients.

Now we prove our result on the truncation error based on Taylor series expansion:
\begin{restatable}{theorem}{taylor}\label{thm:taylor}
	Let $b>0$ and $f(x_0+x)=\sum_{k=0}^\infty a_k x^k$ for every $x\in(-b,b)$ and suppose $\sum_{k=0}^\infty |a_k|b^k\leq B$. Then $g(y):=f(x_0+\frac{2b}{\pi}\arcsin(y))=\sum_{k=0}^{\infty}c_k y^k$ is such that $\sum_{k=0}^{\infty}|c_k|\leq B$, thus for all $\nu,\delta\in (0,1]$ there is a polynomial $P(y)$ of degree $\bigO{\ln(B/\delta)/\nu}$ such that for all $y\in[-1+\nu, 1-\nu]\colon$
\begin{align*}
	\left|g(y)-P(y)\right|\leq \delta
\end{align*}
and for all $y\in [-1,1]$ we have that $|P(y)|$ is bounded by $\delta+\max_{x\in[-\arcsin(1-\nu/2),\arcsin(1-\nu/2)]}|f(x_0+\frac{2b}{\pi}x)|$.
\end{restatable}
\begin{proof}
	The proof is inspired by \cite[Lemma 37]{apeldoorn2017QSDPSolvers} where it is noted that $\vertiii{\frac{2}{\pi}\arcsin(y)}_1=1$. This implies that 
	\begin{align*}
		\vertiii{g(y)}_1
		&=\vertiii{f(x_0+\frac{2b}{\pi}\arcsin(y))}_1\\
		&=\vertiii{\sum_{k=0}^\infty a_k \left(\frac{2b}{\pi}\arcsin(y))\right)^{\!\!k}}_1\\
		&\leq\sum_{k=0}^\infty |a_k| \vertiii{\!\left(\frac{2b}{\pi}\arcsin(y))\right)^{\!\!k}}_1\\
		&\leq\sum_{k=0}^\infty |a_k| \vertiii{\frac{2b}{\pi}\arcsin(y))}_1^{k}\\			
		&=\sum_{k=0}^\infty |a_k| b^k\\
		&\leq B.
	\end{align*}
    Now we can apply \cite[Corollary 66]{gilyen2018quantum} (setting therein $f\leftarrow g,x\leftarrow y,x_0\leftarrow0,r\leftarrow1-\nu,\delta\leftarrow\nu,\eps\leftarrow\delta$) to convert it to a bounded polynomial $P(y)$ on $[-1,1]$.
\end{proof}
Using this theorem we can give analytical bounds on the degree required for approximating the standard normal distribution as follows:
\begin{corollary}\label{cor:ExpDeg}
	Let $\beta,\delta>0$, then there is a degree $ d =\bigO{\beta+\ln(1/\delta)}$ polynomial $P(y)$ bounded by $1$ on $[-1,1]$ such that for every $y\in [-\sin(1),\sin(1)]$ we have that
	\begin{align}\label{eq:normalTruncated}
		\left|\exp\left(-\frac{\beta}{2}\arcsin^2(y)\right)-P(y)\right|\leq \delta.
	\end{align}
\end{corollary}
\begin{proof}
	Apply \Cref{thm:taylor} with setting $b=\frac{\pi}{2}$ and $\nu=1-\sin(1)$ observing that 
	$\exp(-\frac{\beta}{2}x^2)=\sum_{k=0}^\infty \left(-\frac{\beta}{2}\right)^{\!k}\frac{x^{2k}}{k!}$
	and 
	$\sum_{k=0}^\infty \left(\frac{\beta}{2}\right)^{\!k}\frac{\left(\frac{\pi}{2}\right)^{2k}}{k!}
	=\sum_{k=0}^\infty \frac{\left(\frac{\beta\pi^2}{8}\right)^{\!k}}{k!}=\exp\left(\frac{\beta\pi^2}{8}\right)=:B$.
\end{proof}

Similarly, we get analytical bounds on the degree required for approximating the Kaiser window function $W_\beta(x) =\frac{I_0(\beta\sqrt{1-x^2})}{I_0(\beta)}$:
\begin{corollary}\label{cor:KaisDeg}
	Let $\beta,\delta>0$, then there is a degree $ d =\bigO{\beta+\ln(1/\delta)}$ polynomial $P(y)$ bounded by $1$ on $[-1,1]$ such that for every $y\in [-\sin(1),\sin(1)]$ we have that 
	\begin{align}\label{eq:KaiserTruncated}
		\left|W_\beta(\arcsin(y))-P(y)\right|\leq \delta.
	\end{align}
\end{corollary}
\begin{proof}
	We will apply \Cref{thm:taylor} with setting $b=\frac{\pi}{2}$. To compute an upper bound $B$ we observe that
	the smallest possible value of $B$ is given by $\vertiii{W_\beta\left(\frac{\pi}{2}x\right)}_1$. To analyze this quantity let us recall \Cref{eq:I0Taylor} stating that $I_0(z)=G(z^2)$ for the entire
	function $G(z)=\sum_{k=0}^\infty \frac{(z/4)^k}{(k!)^2}$. This means that $W_\beta(x)=\frac{G(\beta^2(1-x^2))}{G(\beta^2)}$ so
	\begin{align}
		\vertiii{W_\beta\left(\frac{\pi}{2}x\right)}_1
		&=\vertiii{G\left(\beta^2(1-\frac{\pi^2}{4}x^2)\right)}_1/G(\beta^2)\\
		&=\vertiii{\sum_{k=0}^\infty \frac{(\beta^2(1-\frac{\pi^2}{4}x^2)/4)^k}{(k!)^2}}_1/G(\beta^2)\\	
		&\leq\sum_{k=0}^\infty \frac{\vertiii{\beta^2(1-\frac{\pi^2}{4}x^2)/4}_1^k}{(k!)^2}/G(\beta^2)\\				
		&=\sum_{k=0}^\infty \frac{(\beta^2(1+\frac{\pi^2}{4})/4)^k}{(k!)^2}/G(\beta^2)\\		
		&\leq\sum_{k=0}^\infty \frac{((4\beta^2)/4)^k}{(k!)^2}/G(\beta^2)\\		
		&=\frac{G\left(4\beta^2\right)}{G(\beta^2)}
		=\frac{I_0\left(2\beta\right)}{I_0\left(\beta\right)}
		\leq\exp\left(\beta\right),	
	\end{align}
	where the last inequality follows from the integral representation of Bessel functions \cite[Eq. (9.6.19)]{abramowitz1974HandbookMathFuns}:
	\begin{align*}
		I_0(2\beta)&=\frac{1}{\pi}\int_0^\pi\exp(2\beta \cos(\theta))d\theta\\
		&\leq\frac{1}{\pi}\int_0^\pi\exp(\beta \cos(\theta))\exp(\beta)d\theta\\
		&=I_0(\beta)\exp(\beta).	\tag*{\qedhere}
	\end{align*}
\end{proof}

Note that the above proofs are constructive in the sense that they also enable explicitly computing approximating polynomials by (approximately) computing the coefficients of the Taylor series. Those coefficients can be computed for example utilizing the Taylor series of $\arcsin(x)=\sum_{\ell=0}^\infty\binom{2\ell}{\ell}\frac{2^{-2\ell}}{2\ell+1}x^{2\ell+1}$.


\section{Analysis of filling fractions}\label{App:Filling}


\begin{lemma}\label{lem:contFill} 
	Consider the functions $\exp(-\frac{\beta}{2}x^2)$ and $W_\beta(x)$ on the interval $[-1,1]$ for some $\beta \geq 0$, then $f(x)\geq 1- \frac{\beta}{2}x^2$ and so the filling fraction satisfies 
	\begin{align}
		\fillfrac{f}{\infty}&\geq 
		\begin{cases}
			\frac{1}{\sqrt[2]{2}}, & \text{for $\beta\leq 2$}\\
			\frac{1}{\sqrt[4]{2\beta}}, & \text{for $\beta\geq 2$}.
		\end{cases}
	\end{align}
\end{lemma}
\begin{proof}
	Since $\exp(x)$ is a convex function we have $\exp(x)\geq 1+x$ and thus $\exp(-\frac{\beta}{2}x^2)\geq 1- \frac{\beta}{2}x^2$. 
	
	Now we prove that $W_\beta(x) \geq 1-\beta x^2/2$ by observing that both functions are even, and they take value 1 at $x=0$. Thus, for showing the inequality it suffices to show that $K'_\beta(x) \geq -\beta x$ for every $x \in (0,1)$. We have that
	\begin{align*}
		K'_\beta(x) &= -\beta x\frac{I_1(\beta \sqrt{1-x^2})}{\sqrt{1-x^2}I_0(\beta)},
	\end{align*}
	so it suffices to show that $g(y):=\frac{I_1(\beta y)}{y I_0(\beta)} \leq 1$ for $y \in
	(0,1)$. Now $g(1)=I_1(\beta)/I_0(\beta)\leq 1$ where the inequality follows from the integral representation of Bessel functions \eqref{eq:I0CosIntegral}. So it suffices to show that $g'(y)\geq 0$ for $y \in (0,1)$. As $g'(y)=(\beta I_2(\beta y))/(y I_0(\beta))$, 
	this holds since both $\beta /(y I_0(\beta)) \geq 0$ and $I_2(\beta y) \geq 0$ follows from \eqref{eq:I0SinIntegral}.
	
	If $\beta\leq 2$ it follows that $\int_{-1}^1 f(x)^2 dx \geq \int_{-1}^1 (1- \frac{\beta}{2}x^2)^2 dx=2-\frac23 \beta +\frac{\beta^2}{10}> 1$,
	and if $\beta\geq 2$ it follows that $\int_{-1}^1 f(x)^2 dx \geq \int_{-\sqrt{\frac{2}{\beta}}}^{\sqrt{\frac{2}{\beta}}} (1- \frac{\beta}{2}x^2)^2 dx=\frac{16}{15}\sqrt{\frac{2}{\beta}}\geq\sqrt{\frac{2}{\beta}}$.
\end{proof}

\begin{lemma}\label{lem:discFill} 
	Let $\beta \geq 0$ and let $f(x)$ be either $\exp(-\frac{\beta}{2}x^2)$ or $W_\beta(x)$. If $N\geq \sqrt{\beta}$ and $|\tilde{f}(\bar{x})-f(\bar{x})|\leq \frac{1}{4}$ for all discrete evaluation points $\bar{x}$ then we have $\fillfrac{\tilde{f}}{N}=\Omega(\frac{1}{\sqrt[4]{\beta+1}})$.
\end{lemma}
\begin{proof}
	Consider the interval $I=[-\frac{1}{\sqrt{\beta}},\frac{1}{\sqrt{\beta}}]\cap[-1,1]$. For all $\bar{x}\in I$ we have $\tilde{f}(\bar{x})\geq f(\bar{x})-\frac{1}{4}\geq \frac{3}{4}- \frac{\beta}{2}\bar{x}^2\geq \frac{1}{4}$, where the first inequality comes from \Cref{lem:contFill}. Therefore,
	\begin{align*}
		\left(\fillfrac{\tilde{f}}{N}\right)^{\!2} &
		\geq \frac{\sum_{\bar{x}\in I} |\tilde{f}(\bar{x})|^2}{N |\tilde{f}|^2_{\mathrm{max}} }
		\geq \frac{\sum_{\bar{x}\in I}(\frac{1}{4})^2}{N (\frac{5}{4})^2}
		=\Omega\left(\frac{1}{\sqrt{\beta+1}}\right).\qedhere
	\end{align*}
\end{proof}


\section{Asymptotic analysis of Gaussian and Kaiser-window state preparation}\label{App:GaussianComplexity}
Here we prove \Cref{thm:GaussianComplexity}, which we restate below:
\GaussianComplexity*
\begin{proof}
	This follows from \Cref{thm:genericComplexity}. 
    For applying this general result we first invoke our filling-fraction bounds \Cref{lem:contFill} and \Cref{lem:discFill} ensuring that $\mathrm{Min}\left(\fillfrac{f}{N},  \fillfrac{\tilde{f}}{N} \right)=\Omega(\frac{1}{\sqrt[4]{\beta+1}})$. Then \Cref{cor:ExpDeg} and \Cref{cor:KaisDeg} implies that we can find a degree $\bigO{\beta + \log(1/\delta)}$ approximating polynomial that has accuracy $\delta = \eps\cdot\mathrm{Min}\left(\fillfrac{f}{N},   \fillfrac{\tilde{f}}{N} \right)=\Omega(\frac{\eps}{\sqrt[4]{\beta+1}})$, proving \Cref{eq:KaiserComplexity}.
	
	Then \Cref{eq:GaussianComplexity} follows from \Cref{eq:KaiserComplexity} by observing that the function $\exp(-\frac{\beta}{2}x^2)$ is $0.5\eps\cdot\mathrm{Min}\left(\fillfrac{f}{N},   \fillfrac{\tilde{f}}{N} \right)$-close to $0$ for $x\gg \sqrt{\frac{2}{\beta}\ln\left(\frac{\sqrt[4]{\beta+1}}{\eps}\right)}=\bigO{\sqrt{\frac{\log(1/\eps)}{\beta}}}$ so we can assume without loss of generality that our approximation $\tilde{f}(x)$ is $0$ for such large values. But then the task reduces to preparing a Gaussian state with $\beta'=\Theta(\log(1/\eps))$ after rescaling $x\rightarrow x'$ so that $x\approx x' \cdot \sqrt{\frac{\log(1/\eps)}{\beta}}$ and adjusting the value of $N'$ appropriately. Note that by choosing the constants appropriately we can even ensure that $N'$ remains a power of $2$.
\end{proof}

\section{Resource estimation details}\label{App:ResourceEst}

In this appendix, we detail the compilation steps used in our resource estimates. We work within a standard fault-tolerant cost model, where the cost of Clifford gates are dominated by the cost of non-Clifford gates, and so we only count the latter. 

\subsection{Resource estimates for QSVT-based method}
The dominant non-Clifford cost in our method is contributed by the $Z$ rotations within $U_{\mathrm{sin}}$. For a degree $d$ approximation polynomial, and a circuit that requires $R$ rounds of amplitude amplification, these gates contribute
\begin{equation}
    (2R+1)d(n+1)
\end{equation} 
$Z$ rotations. Each rotation must be distilled from a number of $T$ gates. Using the approaches of~\cite{kliuchnikov2022shorter}, we can synthesize a $Z$ rotation to diamond-norm error $\delta$ using $0.57 \log_2(1/\delta) + 8.83$ gates. Assuming that these errors add linearly, we synthesize each $Z$ rotation to accuracy $\delta = \epsilon_s / (2R+1)d(n+1)$, where $\epsilon_s$ is the desired error in the final state from rotation synthesis (note that if using our method as a subroutine within an algorithm with additional rotation gates, the value of $\epsilon_s$ will be further reduced to bound the synthesis error in the entire circuit). The $T$ cost is then
\begin{equation}\label{Eq:ResourceEst}
    (2R+1)d(n+1)(0.57 \log_2((2R+1)d(n+1)/\epsilon_s) + 8.83).
\end{equation} 
A small number of additional non-Clifford gates are contributed by the QSVT rotation gates, and the rotation and reflection gates used in amplitude amplification. The QSVT rotations have a factor $(n+1)$ smaller contribution, while the rotation and reflection gates used in amplitude amplification have factor $(n+1)d$ and $d$ smaller contributions, respectively.

\subsection{Resource estimates for amplitude oracles}

The resources required to realize the piecewise-polynomial amplitude oracle are reproduced from Ref.~\cite[Table II]{haner2018optimizing}. For a gaussian function with $L_\infty$-error $\leq 10^{-7}$, the piecewise polynomial approach requires $20,504$~Toffoli gates per oracle call~\cite{haner2018optimizing} and uses $162$ ancilla qubits\footnote{The qubit counts in Table~II of~\cite{haner2018optimizing} are missing one qubit.}.

The resources required to realize the bespoke gaussian amplitude oracle are reproduced from Ref.~\cite[Table II]{poirier2021efficient}, using the `space saving, $0 \leq x' \leq 10$' row. This gives $7546$ Toffoli gates and $133$ qubits. We note that the estimates for the bespoke gaussian amplitude oracle are an optimistic lower bound, as the resource estimates available in~\cite{poirier2021efficient} consider $n=13$ (note that $n$ in this work corresponds to $d$ in Ref.~\cite{poirier2021efficient}), and target a more peaked gaussian with $\beta=100$ (which results in a lower cost than $\beta=10$).

The resources to realize the amplitude oracle for a gaussian function via linear interpolation are optimized using the methods of Refs.~\cite{sanders2020Optimization, berry2023QuantumSimPseudo}. We require approximately $4069$ Toffoli gates and $181$ ancilla qubits.

The approach of Ref.~\cite{sanders2020Optimization} can be viewed as a highly streamlined version of the piecewise-polynomial approach~\cite{haner2018optimizing}. The function to approximate can be divided into a number of intervals, where we perform a separate linear approximation to the function in each interval. The classically computed linear interpolation coefficients (a gradient and intercept) can be coherently loaded for each interval using quantum read-only memory (QROM)~\cite{babbush2018encoding}, where the value of the register storing $\ket{x}$ acts as the address qubits. This approach is refined for the function $f(x) = e^{-x}$ in Ref.~\cite{berry2023QuantumSimPseudo}, by observing that $e^{-x} = 2^{-z}$ where $z = x/\ln(2)$. The efficiency of computing $2^{-z}$ can be improved by exploiting that $2^{-z} = 2^{-z_{\mathrm{int}}}2^{-z_{\mathrm{frac}}}$, where $\mathrm{int}$ and $\mathrm{frac}$ respectively denote the binary integer and fractional parts of the number. Multiplying by $2^{-z_{\mathrm{int}}}$ can be implemented using controlled bit-shift operations. Hence, it is only necessary to perform a linear interpolation for $2^{-z}$ for $0 \leq z \leq 1$, which can be done with few intervals using the interval spacing of Ref.~\cite{sanders2020Optimization}.

The steps considered are listed below:
\begin{enumerate}
    \item Compute $\ket{x} \ket{0} \rightarrow \ket{x} \ket{\sqrt{\frac{10}{\ln(2)}}x} $.
    \item Compute $\ket{\sqrt{\frac{10}{\ln(2)}}x} \ket{0} \rightarrow \ket{\sqrt{\frac{10}{\ln(2)}}x} \ket{z = \frac{10}{\ln(2)}x^2} $.
    \item Using QROM controlled on $z_h$, the high fractional bits of $z$ (see Refs.~\cite{sanders2020Optimization,berry2023QuantumSimPseudo}), load gradients $m_{z_h}$ and intercepts $c_{z_h}$ for each of $g$ intervals. This requires $g$ Toffoli gates and $\log_2(g)$ ancilla qubits.
    \item Compute the linear interpolation to $2^{-z_{\mathrm{frac}}}$ using one multiplication and one addition (we ignore the cost of the addition in this work).
    \item Apply in-place controlled-bit shifts to multiply by $2^{-z_{\mathrm{int}}}$. 
\end{enumerate}
We find numerically that $g=1900$ intervals suffices to achieve $L_\infty$-error $\leq 10^{-7}$ for a linear interpolation of $2^{-z}$ for $0 \leq z \leq 1$. To store the output of the amplitude oracle to $L_\infty$-error $\leq 10^{-7}$ requires 24 qubits. We use 2 integer and 27 fractional bits for the output register in step 1) above. We use 3 integer bits and 27 fractional bits for the output register in step 2) above. We use 24 bits of each of the registers storing the gradient and intercept in step 3). Finally we use 24 bits for the output register used in steps 4) and 5). The most ancilla qubits required in a step is approximately 50, for the multiplication in step 4). We reuse these during the other steps. The total ancilla count for the linear interpolation amplitude oracle for the Gaussian is then approximately $181$. We note that this could be reduced by uncomputing and reusing work registers, or by using ancilla-free multiplication algorithms. However, this may increase the gate count, and we do not explore these optimizations here.

The gate count depends sensitively on the cost of quantum multiplication, which we treat as roughly $n_b^\alpha \times n_b^\beta$ here, where $n_b^{\alpha/\beta}$ is the number of binary digits used to store each of the numbers. The multiplication in step 1) costs approximately $16 \times 29 = 464$ Toffolis. The squaring in step 2) costs approximately $29^2 = 841$ Toffolis. Loading the QROM in step 3) costs $1900$ Toffolis. The multiplication in step 4) costs approximately $576$ Toffolis. The controlled bit-shift in step 5) costs approximately $288$ Toffolis~\cite{berry2023QuantumSimPseudo}. This gives a total gate count of $4069$ Toffoli gates.

\end{document}